\documentclass[12pt]{article}

\usepackage[margin=1in]{geometry}
\usepackage{amsmath,amsthm,amssymb,mathtools}
\usepackage{bm}
\usepackage{graphicx}
\usepackage{color}
\usepackage{bbm}

\newtheorem{theorem}{Theorem}
\newtheorem{lemma}{Lemma}

\newtheorem{proposition}{Proposition}

\newcommand{\Z}{\mathbb{Z}}

\newcommand{\R}{\mathbb{R}}

\newcommand{\T}{\mathbb{T}}

\newcommand{\Cov}{\mathrm{Cov}}

\newcommand{\mcA}{{A}}

\newcommand{\mcM}{{M}}

\newcommand{\Mt}{\widetilde{M}}

\newcommand{\sbl}{[\,}
\newcommand{\sbr}{\,]}
\newcommand{\sqb}[1]{[\, #1 \,]}

\usepackage[round]{natbib}

\begin{document}

\begin{center}

{\LARGE Inverses of Mat\'ern Covariances on Grids}

\vspace{16pt}

{Joseph Guinness}
\vspace{8pt}

\textit{Cornell University, Department of Statistics and Data Science}
\vspace{16pt}

\textbf{Abstract}
\end{center}

We conduct a study of the aliased spectral densities of Mat\'ern
covariance functions on a regular grid of points, providing clarity on
the properties of a popular approximation based on stochastic partial
differential equations; while others have shown that it can
approximate the covariance function well, we find that it assigns too
much power at high frequencies and does not provide increasingly
accurate approximations to the inverse as the grid spacing goes to
zero, except in the one-dimensional exponential covariance case. We
provide numerical results to support our theory, and in a simulation
study, we investigate the implications for parameter estimation,
finding that the SPDE approximation tends to overestimate spatial
range parameters.

\section{Introduction}\label{introduction}

The Mat\'ern covariance between two points in $\R^d$ separated by lag $h \in \R^d$ is
\begin{align}
M \sqb{ h \,;\, \nu, d } = 
\frac{\sigma^2}{2^{\nu-1} \Gamma(\nu)}(\alpha \| h\|)^\nu K_\nu(\alpha \| h \|),
\end{align}
where $\sigma^2$ is a variance parameter, $\alpha$ is an inverse range
parameter, $\nu$ is a smoothness parameter, and $K_\nu$ is the
modified Bessel function of the second kind.
\cite{guttorp2006studies} provide a summary of its important
properties and a detailed discussion of its history.  
Our article presents a theoretical and numerical study of properties of the
spectral density of the Mat\'ern covariance when aliased to regular
grids of points in one and two dimensions. We apply our results to
study a popular approximation to the inverse of Mat\'ern covariance
matrices that is motivated by connections between the Mat\'ern covariance and
a class of stochastic partial differential equations (SPDEs)
\citep{lindgren2011explicit}.

Building on work by \cite{whittle1954stationary},
\cite{whittle1963stochastic}, and \cite{besag1981system},
\cite{lindgren2011explicit} proposed that the inverse of Mat\'ern
covariance matrices can be represented by sparse matrices whenever
$\nu + d/2$ is an integer, which is why our notation for $M$
includes $\nu$ and $d$. The resulting approximation is commonly
referred to as the SPDE approximation. In this paper, we use the terms
``SPDE approach'' and ``SPDE approximation'' to refer specificially to
the methods in \cite{lindgren2011explicit}. We investigate the sparsity
of Mat\'ern inverses and find that there is nothing particularly
special with regards to sparsity about the $d=1, \nu = 3/2$ case or
the $d=2, \nu = 1$ case, relative to other values of $\nu$. Further,
by studying the spectral densities implied by the SPDE approximation,
we show that the SPDE over-approximates power at the highest
frequencies by a factor of 3 in the $d=1$, $\nu = 3/2$ case and by as
much as a factor of 2.7 in the $d=2$, $\nu = 1$ case. 

In the discussion of \cite{lindgren2011explicit}, Lee and Kaufman
noted that the likelihood implied by the SPDE approximation
overestimates spatial range parameters. The discussion of the bias was
centered on boundary effects. Though boundary effects are important
for approximations to the inverse covariance matrix, the present paper
suggests instead that the overestimation stems from the fact that the
SPDE approximation has too much power at the highest frequencies,
causing the likelihood to select a larger range parameter to
compensate. Section \ref{numerical_section} of the present paper
contains a simulation study that corroborates Lee and Kaufman's
results. Our results also suggest an explanation for why
\cite{guinness2018permutation} found that SPDE approximations were
less accurate in terms of KL-divergence than Vecchia's approximation
\citep{vecchia1988estimation}.

To set the stage, consider a Mat\'ern covariance matrix $\Sigma$ for a
large grid of points in one dimension with spacing 1, ordered left to
right. Then let $Q = \Sigma^{-1}$, and consider the values
$Q_{i,i+h}/Q_{i,i}$, where $i$ is the index for a location near the
center of the domain. Figure 1 plots these values for $h = 1, 2, 3$,
and $4$ over a range of smoothness and inverse range parameters. When
$\nu = 0.5$, the values are zero for $h > 1$, and they generally
appear to converge as the inverse range decreases, but not to zero
when $\nu = 3/2$, which is the approximation used in the SPDE
approach. The rest of the present paper aims to explore
properties of the Mat\'ern model and the SPDE approximation, with an
aim of understanding these numerical results.

\begin{figure}
  \centering
  \includegraphics[width=\textwidth]{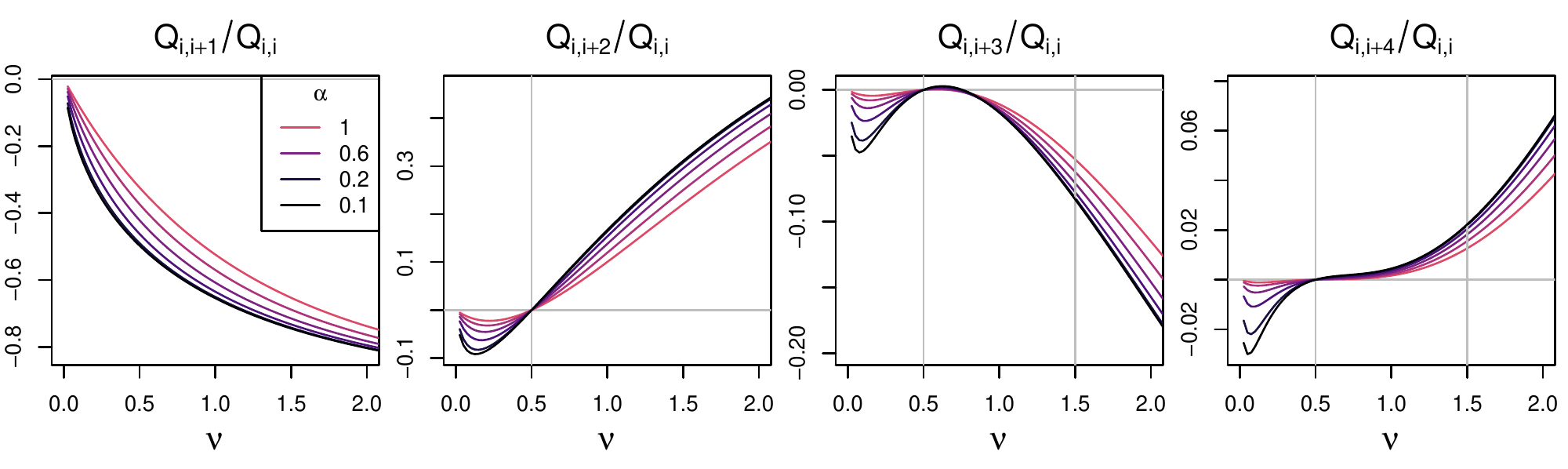}
  \caption{\label{Q_vs_nu_1d_intro} For dimension $d=1$, $Q_{i,i+h}/Q_{i,i}$
    as a function of $\nu$ for various values of $h$ and several
    inverse range parameters $\alpha$. Vertical lines indicate values
    set to zero in the SPDE approximation. }
\end{figure}

Section \ref{background} contains a general background on spectral
theory for stationary random fields on grids of points. Section
\ref{matern} provides a theoretical study of spectral properties of
Mat\'ern covariances in particular, and of the SPDE
approximation. Section \ref{numerical_section} includes numerical and
simulation results, and Section \ref{discussion} concludes with a
discussion. The appendix has an extended background and proofs of
theorems.

\section{Background}\label{background}

The details for all derivations in this section are spelled out in
Appendix \ref{spectral_appendix}.  Let $Y : \R^d \to \R$ be a
stationary process with autocovariance function
$A\sbl h \sbr = \Cov \{ Y\sqb{x+h},Y\sqb{x} \}$. Due to Bochner's
theorem \citep[cf.][]{stein1999interpolation}, $A\sbl \sbr$ is
positive definite when
\begin{align}
\int_{\R^d} A \sbl h \sbr e^{-i 2\pi \omega \cdot h} dh
 \coloneqq  {A}(\omega) > 0 \quad \mbox{for all } \omega \in \R^d.
\end{align}
We call $\mcA()$ the spectral density for $A \sbl \sbr$. Our
notational convention uses the same letter for the spectral density
and covariance function, distinguishing the two with the type of
bracket: round for spectral densities and square for covariances. For
$\Delta > 0$, define the interval $\T_\Delta = [0,1/\Delta]$ and
hypercube $\T_\Delta^d$. When $h \in \Z^d$, the inverse Fourier
transform can be rewritten as
\begin{align}
  A \sbl \Delta h \sbr
  = \int_{\T_\Delta^d}  \,\,  \sum_{k \in \Z^d} 
  \mcA(\omega + k/\Delta)e^{i2\pi \Delta \omega \cdot h}d\omega
  \eqqcolon A_\Delta \sqb{h},  
\end{align}
which uses the aliasing property of complex exponentials and
introduces a notation $A_\Delta\sqb{} : \Z^d \to \R$ for covariances on a
grid of points with spacing $\Delta$. We define
\begin{align}
\mcA_\Delta(\omega) = \sum_{k \in \Z^d} \mcA( \omega + k/\Delta )
\end{align}
to be the aliased spectral density for $A$ on a grid with spacing
$\Delta$. The discrete covariances and the aliased spectral density
are related via
\begin{align}
  A_\Delta \sqb{h} =
  \int_{\T_\Delta^d}
  A_\Delta(\omega) e^{i2\pi\Delta \omega \cdot h} d\omega , 
  \quad
  A_\Delta(\omega) = \Delta^d 
  \sum_{h \in \Z^d} A_\Delta \sqb{h} e^{-i2\pi\Delta \omega \cdot h},
\end{align}
so that $A_\Delta\sqb{}$ is the integral Fourier transform of
$A_\Delta()$ over $\T_\Delta^d$, and $A_\Delta()$ is the infinite
discrete Fourier transform of $A_\Delta \sqb{}$. We say that
$A_\Delta^{-1}\sqb{}$ is the inverse of $A_\Delta\sqb{}$ if
\begin{align}\label{inverse_def}
 \Delta^d \sum_{k \in \Z^d} A_\Delta \sbl h-k \sbr A_\Delta^{-1} \sbl k \sbr 
= \mathbbm{1} \sbl h \sbr.
\end{align}
where $\mathbbm{1}\sbl h\sbr = 1$ when $h=0$ and $0$ otherwise. We
call $A_\Delta^{-1}\sqb{}$ the inverse operator.  Taking the infinite
discrete Fourier transform of both sides of \eqref{inverse_def}
reveals that
\begin{align}
  A_\Delta(\omega)A^{-1}_\Delta(\omega) = \Delta^d,
\end{align}
meaning that the spectrum of $A_\Delta^{-1}$ is the $\Delta^d$ times
the reciprocal of the spectrum of $A_\Delta$.

We can also define the square root of $A_\Delta$ to be the operator
$A^{1/2}_\Delta$ for which
\begin{align}\label{opmult_transpose}
A_\Delta \sbl h \sbr = 
\Delta^d \sum_{k \in \Z^d}  A^{1/2}_\Delta \sqb{h-k}A^{1/2}_\Delta \sqb{-k}
\end{align}
Note the difference between \eqref{inverse_def}, which is meant to mimic
the matrix multiplication $B B^{-1}$, and \eqref{opmult_transpose}, which
is meant to mimic the matrix multiplication $B B^T$.  Taking the
Fourier transform of both sides reveals that
\begin{align}
A_\Delta^{1/2}(\omega) A_\Delta^{1/2}(\omega)^* = A_\Delta(\omega),
\end{align}
where $*$ denotes complex conjugate.  The spectral density of
$A_\Delta^{1/2}$ is the complex square root of the spectral density of
$A_\Delta$. This means that the square root spectral density is unique
only up to multiplication by $\exp(i 2\pi \omega x)$. This is a necessary
consequence of the translation-invariance property of stationary
processes, since multiplication by $\exp(i 2\pi\omega x)$ in the spectral
domain corresponds to translation by $x$ in the natural domain. Square
root operators may also have inverses; their spectral densities again
follow the reciprocal relationship.

Square root operators are useful for the simulation of processes that
have particular covariances. Let $W: \Z^d \to \R$ be a white noise
process defined on a the integer lattice, that is, its autocovariance
function is the identity function $\mathbbm{1}\sqb{h}$. Define
$Y: (\Delta \Z)^d \to \R$ as
\begin{align}\label{convolution}
Y \sqb{\Delta j} = \Delta^{d/2} \sum_{k \in \Z^d} A_\Delta^{1/2} \sqb{ j-k} W\sqb{k}.
\end{align}
The covariance function for $Y$ is $A_\Delta$. The representation in
\eqref{convolution} is exploited in the convolution method of
\cite{higdon1998process}.  The inverse of the square root operator can
be used to decorrelate a process. For a process $Y$ on $(\Delta \Z)^d$
with autocovariance function $A_\Delta$, let $A_\Delta^{-1/2}$ be the
square root operator for $A_\Delta^{-1}$, and define $W$ on $\Z^d$ as
\begin{align}
  W \sqb{j}  = \Delta^{d/2} \sum_{k \in \Z^d} A_\Delta^{-1/2} \sqb{ k-j}  Y \sqb{\Delta k}.
\end{align}
The process $W$ has covariance function $\mathbbm{1}\sqb{h}$ and is
thus a white noise process.

\section{Mat\'ern Covariances and SPDE Approximations}\label{matern}

The stationary Mat\'ern covariance function is
\begin{align}
  M\sqb{ h \,;\, \nu, d}
  = \frac{ \sigma^2 (\alpha \| h \|)^{\nu} K_\nu(\alpha \| h \|)}{\Gamma(\nu)2^{\nu-1}}
  =  \int_{\R^d} \frac{ \sigma^2 N_{\alpha,\nu,d} }
     { \big( \alpha^2 + 4\pi^2 \| \omega \|^2 \big)^{\nu+d/2}}
         e^{i 2\pi \omega \cdot h}d\omega, 
\end{align}
where
$N_{\alpha,\nu,d} = 2^d
\pi^{d/2}\alpha^{2\nu}\Gamma(\nu+d/2)/\Gamma(\nu)$ is a normalizing
constant \citep{williams2006gaussian}. The aliased spectral density is 
\begin{align}
  \mcM_\Delta(\omega\,;\,\nu, d) = \sum_{k \in \Z^d}
  \sigma^2 N_{\alpha,\nu,d}
       \big(\alpha^2 + 4\pi^2 \|\omega + k/\Delta\|^2 \big)^{-\nu - d/2}.
\end{align}

\subsection{One Dimension, $\nu = 1/2$}

From here on, we set $\sigma^2 = 1$ to simplify the expressions. When
$d=1$ and $\nu = 1/2$, the aliased spectral density has the closed
form
\begin{align}\label{specden_12_1}
  M_\Delta(\omega\,;\,1/2,1) 
  = \Delta \frac{ 1 - e^{-2\Delta \alpha}}
                 {1 + e^{-2\Delta \alpha} 
                     - e^{-\Delta \alpha}e^{-i\omega 2\pi\Delta}
                     - e^{-\Delta \alpha}e^{+i\omega 2\pi\Delta} },
\end{align}
which can be proven by taking the discrete Fourier transform of the
covariance function. The inverse spectral density is $\Delta$ times the
reciprocal,
\begin{align}
  M^{-1}_\Delta(\omega\,;\,1/2,1) =
\frac{1 + e^{-2\Delta \alpha}  
        - e^{-\Delta \alpha}e^{-i\omega 2\pi\Delta}
        - e^{-\Delta \alpha}e^{+i\omega 2\pi\Delta} }
     {1 - e^{-2\Delta \alpha}},
\end{align}
and thus the inverse operator is
\begin{align}
  M^{-1}_\Delta\sqb{ h\,;\,1/2,1 } 
= \left\lbrace
  \begin{array}{ll}
    (1 + e^{-2\Delta \alpha})/(1-e^{-2\Delta \alpha}) & h = 0 \\
    - e^{-\Delta \alpha}/(1 - e^{-2\Delta \alpha}) & |h| = 1 \\
    0 & |h| > 1.
  \end{array}
  \right.      
\end{align}
The SPDE approximation for the inverse operator in
\cite{lindgren2011explicit} is
\begin{align}\label{spde_inv_12_1}
   \Mt_\Delta^{-1}\sqb{h\,;\, 1/2, 1}  = \left\lbrace
  \begin{array}{ll}
     \frac{\alpha\Delta}{2} + \frac{1}{\alpha\Delta} & \hspace{1.3mm} h \hspace{1mm}  = 0 \\
     -\frac{1}{2\alpha\Delta}  & |h| = 1 \\
     0 & |h| > 1.
  \end{array} \right. 
\end{align}
which corresponds to spectral density
\begin{align}\label{spde_spec_12_1}
\Mt_\Delta( \omega \,;\, 1/2, 1 ) = \Delta\Big( \frac{\alpha\Delta}{2} + \frac{1}{\alpha\Delta} - \frac{1}{2\alpha\Delta} e^{-i\omega 2\pi \Delta}- \frac{1}{2\alpha\Delta} e^{+i\omega 2\pi \Delta} \Big)^{-1}.
\end{align}
Our first theorem establishes that the true and SPDE spectral
densities for $d=1, \nu=1/2$ converge to the same values at frequencies $0$ and 
$1/(2\Delta)$ for small $\alpha \Delta$.
\begin{theorem}\label{specden_1o2}
\begin{align*}
\frac{M_\Delta( 0 \,;\, 1/2, 1 )}{2/\alpha} &= 1 + O(\alpha^2 \Delta^2)  
&\frac{M_\Delta\big( \frac{1}{2\Delta} \,;\, 1/2, 1 \big)}{2/\alpha} &= \frac{\alpha^2\Delta^2}{4} + O(\alpha^4 \Delta^4) \\
\frac{\Mt_\Delta( 0 \,;\, 1/2, 1 )}{2/\alpha} &=  1 
&\frac{\Mt_\Delta\big( \frac{1}{2\Delta} \,;\, 1/2, 1 \big)}{2/\alpha} &= \frac{\alpha^2\Delta^2}{4} + O(\alpha^4\Delta^4).
\end{align*}
\end{theorem}
Proofs for all theoretical results are in Appendix \ref{proofs}.
Thorem \ref{specden_1o2} provides evidence that the SPDE approximation for $\nu = 1/2$,
$d=1$ is a good approximation to the true model when $\alpha\Delta$ is
small; their spectral densities are similar at the lowest frequency
($\omega = 0$) when the power is greatest and at the highest frequency
($\omega = \Delta^{-1}/2$) when the power is smallest, implying that
the both the SPDE spectral density and its reciprocal may be good
approximations to the truth, which in turn implies that both the
covariance operator and its inverse may be good approximations.

\subsection{One Dimension, $\nu = 3/2$}

When $\nu = 3/2$, the aliased spectral density is
\begin{align}\label{specden_32_1}
  \mcM_\Delta(\omega\,;\, 3/2,1) &= 
  \sum_{k \in \Z} \frac{4 \alpha^3}{[ \alpha^2 + 4\pi^2 (\omega + k/\Delta)^2 ]^2}.
\end{align}
The SPDE approximation in \cite{lindgren2011explicit} to the inverse
operator is simply the convolution of the $\nu = 1/2$ approximation 
\eqref{spde_inv_12_1} with itself
(here, normalizing constants are chosen so that $M_\Delta(0 \,;\, 3/2, 1) \to
\Mt_\Delta(0 \,;\, 3/2, 1) = 4/\alpha$ as $\alpha\Delta \to 0$),
\begin{align}
   (\alpha\Delta) \Mt_\Delta^{-1}\sqb{h\,;\, 3/2, 1}  = \left\lbrace
  \begin{array}{ll}
     \Big(\frac{\alpha\Delta}{2} + \frac{1}{\alpha\Delta}\Big)^2 + 
     \frac{1}{2\alpha^2\Delta^2} 
         & \hspace{1.3mm} h \hspace{1mm}  = 0 \\
     -\frac{1}{2} - \frac{1}{\alpha^2\Delta^2}  
         & |h| = 1 \\
     \frac{1}{4\alpha^2\Delta^2} & |h| = 2 \\
     0 & |h| > 2,
  \end{array} \right. 
\end{align}
which means that the spectral density for the $\nu = 3/2$ SPDE inverse
operator is simply the square of spectral density for the $\nu = 1/2$
SPDE inverse operator \eqref{spde_spec_12_1}, 
\begin{align}\label{spde_inv_32_1}
\Mt_\Delta^{-1}(\omega \,;\, 3/2, 1 ) = 
   \frac{1}{\alpha\Delta}\Big( 
     \frac{\alpha\Delta}{2} 
    +\frac{1}{\alpha\Delta} 
    -\frac{1}{2\alpha\Delta} e^{-i\omega 2\pi \Delta} 
    -\frac{1}{2\alpha\Delta} e^{+i\omega 2\pi \Delta} 
   \Big)^{2},
\end{align}
and the spectral density for the $\nu = 3/2$ 
SPDE covariance operator is
\begin{align}
\Mt_\Delta( \omega \,;\, 3/2, 1 ) = \alpha\Delta^2 
   \Big( 
     \frac{\alpha\Delta}{2} 
    +\frac{1}{\alpha\Delta} 
    -\frac{1}{2\alpha\Delta} e^{-i\omega 2\pi \Delta} 
    -\frac{1}{2\alpha\Delta} e^{+i\omega 2\pi \Delta} 
   \Big)^{-2}.
\end{align}
Note however, that the aliased Mat\'ern spectral density for
$\nu = 3/2$ in \eqref{specden_32_1} is not simply the square of the
aliased $\nu = 1/2$ spectral density in \eqref{specden_12_1}; rather,
we alias the square of the unaliased $\nu = 1/2$ spectral density. The
SPDE appproximation reverses the order of operations, squaring the
aliased spectral density. This subtle difference leads to the SPDE
approximation assigning too much power at the highest frequencies,
made explicit in the following theorem:
\begin{theorem}\label{specden_3o2}
\begin{align*}
\frac{M_\Delta( 0 \,;\, 3/2, 1 )}{4/\alpha} &= 
    1 + O(\alpha^4 \Delta^4) 
&\frac{M_\Delta\big( \frac{1}{2\Delta} \,;\, 3/2, 1 \big)}{4/\alpha} &=
    \frac{\alpha^4\Delta^4}{48} + O(\alpha^6 \Delta^6)  \\ 
\frac{\Mt_\Delta( 0 \,;\, 3/2, 1 )}{4/\alpha} &= 1 
&\frac{\Mt_\Delta\big( \frac{1}{2\Delta} \,;\, 3/2, 1 \big)}{4/\alpha} &=
    \frac{\alpha^4\Delta^4}{16} + O(\alpha^6\Delta^6).
\end{align*}
\end{theorem} 
When scaled by $4/\alpha$, both spectral densities converge to
1 when $\omega = 0$ and $\alpha\Delta \to 0$, but they 
converge to two different values, $\alpha^4\Delta^4/48$ and
$\alpha^4\Delta^4/16$, when $\omega = \Delta^{-1}/2$, meaning that the
SPDE spectral density assigns three times too much power at the
highest frequency.  The inaccuracy of the spectral density at high
frequencies impacts the quality of the approximation to the reciprocal
of the spectral density, seen in Figure \ref{spec_recip}, and to the inverse
operator, as evidenced in Figure \ref{Q_vs_nu_1d_intro}.

     \begin{figure}
       \centering
       \includegraphics[width=0.8\textwidth]{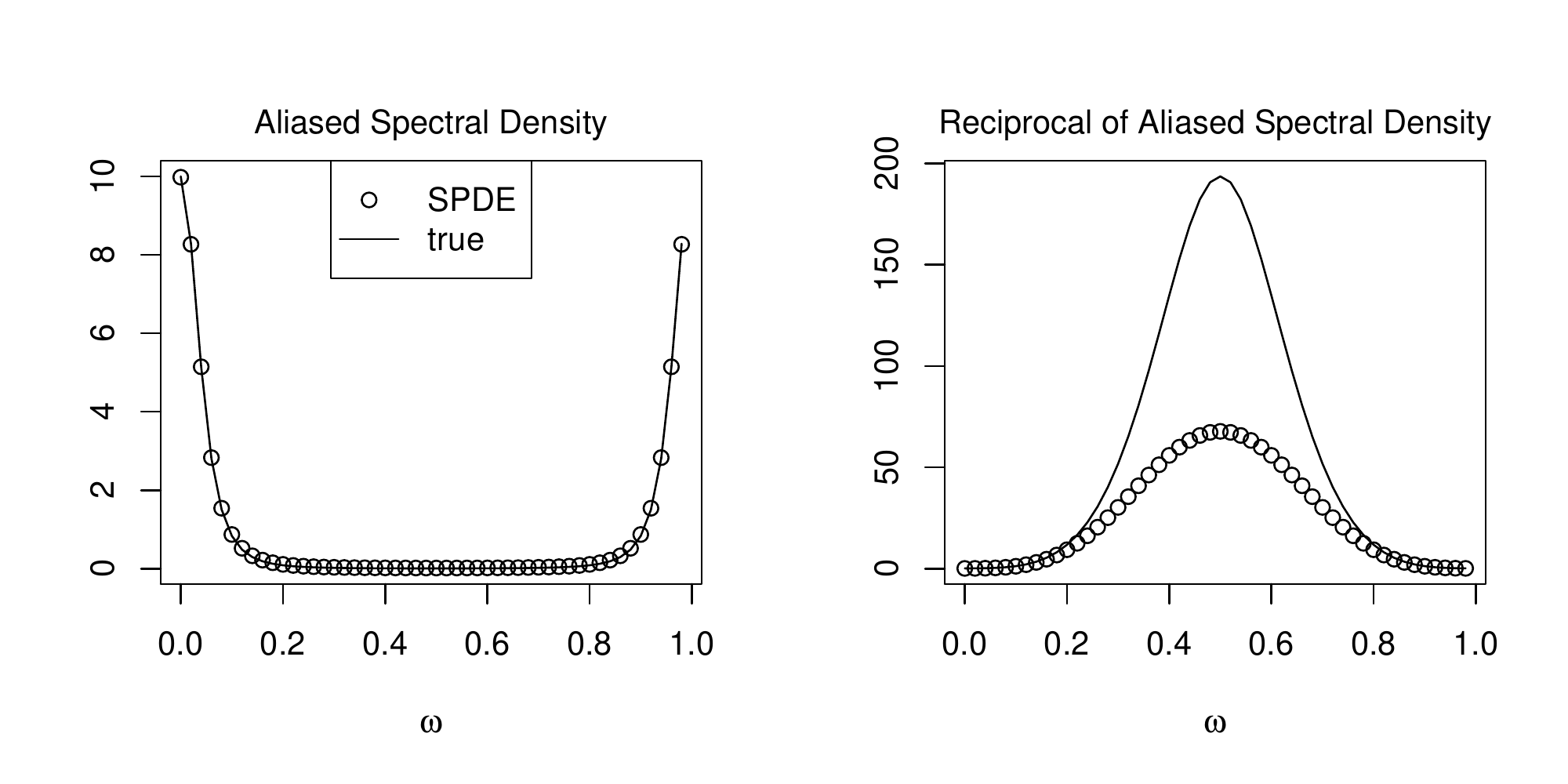}
       \caption{\label{spec_recip} For $d=1$, $\nu = 3/2$,
         $\alpha = 0.4$, true aliased spectral density and its
         reciprocal (lines) and SPDE approximation to the aliased
         spectral density and its reciprocal (circles).}
     \end{figure}

\subsection{Two Dimensions}

The aliased spectral density for the Mat\'ern in two dimensions is
\begin{align}
  M_\Delta(\omega \,;\, \nu, 2 ) = 4\pi \alpha^2 \sum_{k \in \Z^2}
  \big[ \alpha^2 + 
        4\pi^2 (\omega_1 + k_1/\Delta)^2 + 
        4\pi^2 (\omega_2 + k_2/\Delta)^2 
  \big]^{-\nu-1}.
\end{align}
The following theorem establishes properties of the
aliased Mat\'ern spectral density for $\nu = 1$ at the lowest
frequency and at high frequencies in one and both spatial dimensions.
\begin{theorem}\label{theorem_2d_true}
\begin{align}
\frac{M_\Delta( (0,0) \,;\, 1, 2 )}{4\pi/\alpha^2} &=
 1 +  \frac{\alpha^4 \Delta^4}{258.6} + O(\alpha^6 \Delta^6)\\
\frac{ M_\Delta\big( (\frac{1}{2\Delta},0)\,;\, 1,2 \big) }{4\pi/\alpha^2} &=
  \frac{\alpha^4\Delta^4}{43.10} + O(\alpha^6\Delta^6)\\
\frac{ M_\Delta\big( (\frac{1}{2\Delta},\frac{1}{2\Delta})\,;\, 1,2 \big) }{4\pi/\alpha^2} &=
 \frac{\alpha^4\Delta^4}{86.20} + O(\alpha^6\Delta^6)
\end{align}
\end{theorem}
The numbers 258.6, 43.10, and 86.20 are the result of numerical
calculations and are rounded to one or two decimals.  They are
available to higher accuracy. Details are given in the proof in the
supplementary material. The SPDE approximation to the inverse
operator is
\begin{align}\label{SPDE_coefs_2d}
4\pi (\alpha\Delta)^2 \widetilde{M}_\Delta^{-1} \sqb{h\,;\,1,2} = \left\lbrace \begin{array}{rl}
(4 + \alpha^2\Delta^2)^2 + 4 & \quad h = (0,0) \\
-2(4+\alpha^2\Delta^2) & \quad h = (0,1), (0,-1), (1,0), (-1,0) \\
2 & \quad h = (1,1), (1,-1), (-1,1), (-1,-1) \\
1 & \quad h = (2,0), (0,2), (-2,0), (0,-2) \\
0 & \quad \mbox{otherwise},
\end{array}
\right.
\end{align}
which corresponds to the spectral density
\begin{align}
\widetilde{M}_\Delta(\omega \,;\, 1, 2) = 
4\pi\alpha^2\Delta^4
\Big( 4 + \alpha^2\Delta^2 
- e^{i2\pi\Delta\omega_1} 
- e^{-i2\pi\Delta\omega_1} 
- e^{i2\pi\Delta\omega_2} 
- e^{-i2\pi\Delta\omega_2} \Big)^{-2}.
\end{align}
The normalizing constants are chosen so that
$M_\Delta((0,0) \,;\, 1,2 ) \to \Mt_\Delta((0,0) \,;\, 1,2) = 4\pi/\alpha^2$ as
$\alpha\Delta \to 0$. The following theorem establishes the behavior
of $\Mt_\Delta(\omega \,;\, 1, 2)$ at the same frequencies in Theorem \ref{theorem_2d_true}.
\begin{theorem}\label{theorem_2d_SPDE}
\begin{align}
\frac{\Mt_\Delta( (0,0) \,;\, 1, 2 )}{4\pi/\alpha^2} 
&= 1  \\
\frac{\Mt_\Delta\big( (\frac{1}{2\Delta},0)\,;\, 1,2 \big) }{4\pi/\alpha^2} 
&= \frac{\alpha^4\Delta^4}{16}  + O( \alpha^6\Delta^6 )\\
\frac{\Mt_\Delta\big( (\frac{1}{2\Delta},\frac{1}{2\Delta})\,;\, 1,2 \big) }{4\pi/\alpha^2} 
&= \frac{\alpha^4\Delta^4}{64} + O( \alpha^6\Delta^6 )
\end{align}
\end{theorem}

Theorems \ref{theorem_2d_true} and \ref{theorem_2d_SPDE} imply that
when $\alpha\Delta$ is small, the SPDE approach over-approximates the spectral
density by a factor of $43.1/16 = 2.69$ at $\omega = (\Delta^{-1}/2,0)$
and $86.2/64 = 1.35$ at $\omega = (\Delta^{-1}/2,\Delta^{-1}/2)$. Figure
\ref{spec_ratio_2d} contains an example where the ratio between the
SPDE spectral density and the true spectral density varies between
0.999 at $\omega = (0,0)$, 2.680 at $\omega = (\Delta^{-1}/2,0)$ and 1.345
at $\omega = (\Delta^{-1}/2,\Delta^{-1}/2)$.

    \begin{figure}
      \centering
      \includegraphics[width=1.0\textwidth]{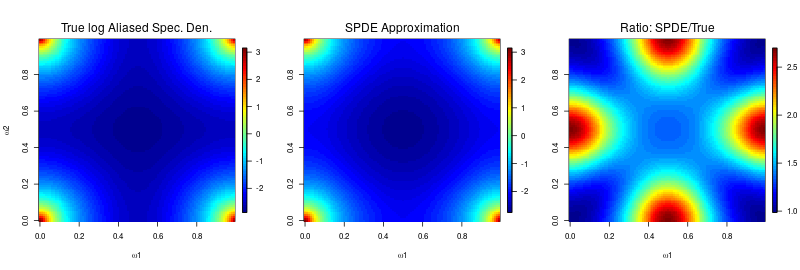}
      \caption{\label{spec_ratio_2d} True spectral density for
        $\nu = 1$, $\alpha = 0.5$, SPDE approximation to the
        spectral density, and the ratio of the two. The
        ratio is near 1.00 at $(0,0)$, near 2.69 at $(1/2\Delta,0)$ and
        near 1.35 at $(1/2\Delta,1/2\Delta)$, as predicted by the theory.}
    \end{figure}

\section{Numerical and Simulation Results}\label{numerical_section}

\subsection{Dependence of Inverse Operator on Smoothness}

In Figures \ref{Q_vs_nu_1d} and \ref{Q_vs_nu_2d}, we plot
$M^{-1}_\Delta\sqb{0 \,;\, \nu,d}$ and
$M^{-1}_\Delta\sqb{h \,;\, \nu, d}/M^{-1}_\Delta\sqb{0 \,;\, \nu, d}$ for
$\Delta = 1$, $d=1,2$, and for a range of values of smoothness
parameter $\nu$ and inverse range $\alpha$. First consider $d=1$. When
$h=1$, the operator is always negative. When $h = 2$, the operator is
exactly 0 for $\nu = 0.5$, agreeing with the theory from Section 3 of
the main document, which says that the operator is exactly zero when
$\nu = 0.5$ for all $|h|>1$. When $h=3$, there is an additional zero
near $\nu = 0.7$ but there is no zero at $\nu = 1.5$; the SPDE
approximation \citep{lindgren2011explicit} sets the operator equal to
zero when $\nu = 1.5$ for all $|h| > 2$. When $h=4$, the only zero is
at $\nu = 0.5$. For $d=2$, the SPDE approximation sets the inverse
operator to zero when $\nu = 1.0$ and $|h_1| + |h_2| > 2$.  Figure
\ref{Q_vs_nu_2d} shows that while there are some zeros in the inverse
operator, they generally do not appear at or near $\nu = 1.0$. For
example, when $h = (1,2)$, the inverse operator is nearly at its
maximum when $\nu = 1.0$. While the magnitudes of the operators
generally decrease as $\|h\|$ increases, there does not appear to be
anything particularly sparse about the cases $d=1, \nu = 3/2$ or
$d=2, \nu = 1$.

\subsection{Simulation Study}

We simulated two-dimensional data on a $(30,30)$ grid under a Mat\'ern
model with $\sigma^2 = 2$, $\alpha = 0.2$, and $\nu = 1$, and with
three noise levels, $\tau^2 = 0$, $\tau^2 = 0.01$, and $\tau^2 =
0.1$. The model is parameterized so that the total variance is
$\sigma^2(1 + \tau^2)$, which means that $1/\tau^2$ is the signal to
noise ratio. Data are simulated using a standard method of forming the
$900 \times 900$ true covariance matrix and multiplying its Cholesky
factor by a vector of standard normals. We assumed that the mean was
known to be zero, and the smoothness parameter $\nu$ was known to be
1. In the zero noise case ($\tau^2 = 0$), we assumed that the noise
parameter was known to be zero; otherwise, we estimated the noise
parameter $\tau^2$ along with $\sigma^2$ and $\alpha$.

\begin{figure}
  \centering
  \includegraphics[width=\textwidth]{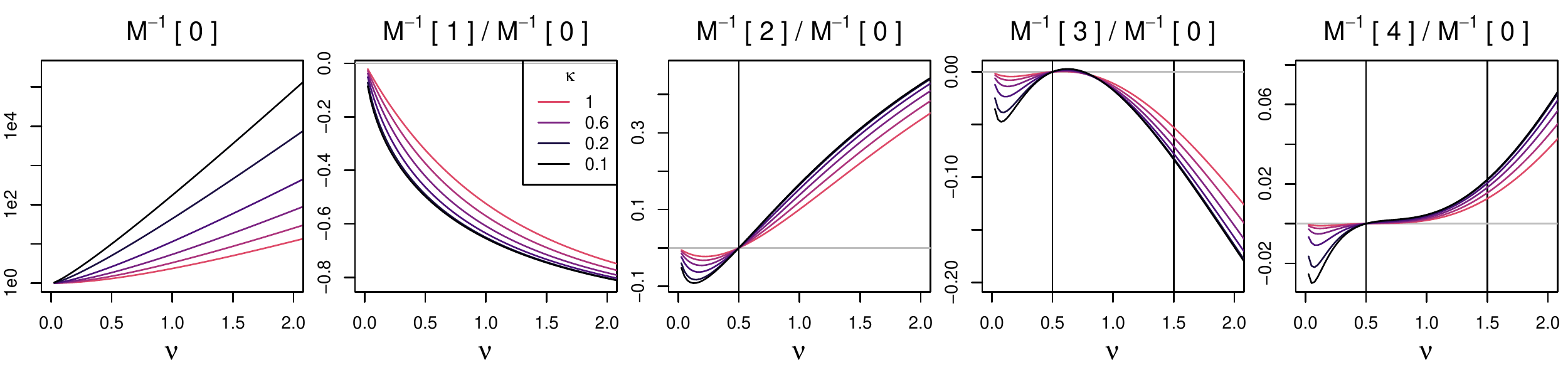}
  \caption{\label{Q_vs_nu_1d} For dimension $d=1$,
    $M_1^{-1}\sqb{h \,;\, \nu, 1}/M_1^{-1}\sqb{0 \,;\, \nu, 1}$ as a function
    of $\nu$ for various values of $h$ and several inverse range
    parameters $\alpha$. Vertical lines indicate values set to zero in
    the SPDE approximation.}
\end{figure}

\begin{figure}
  \centering
  \includegraphics[width=\textwidth]{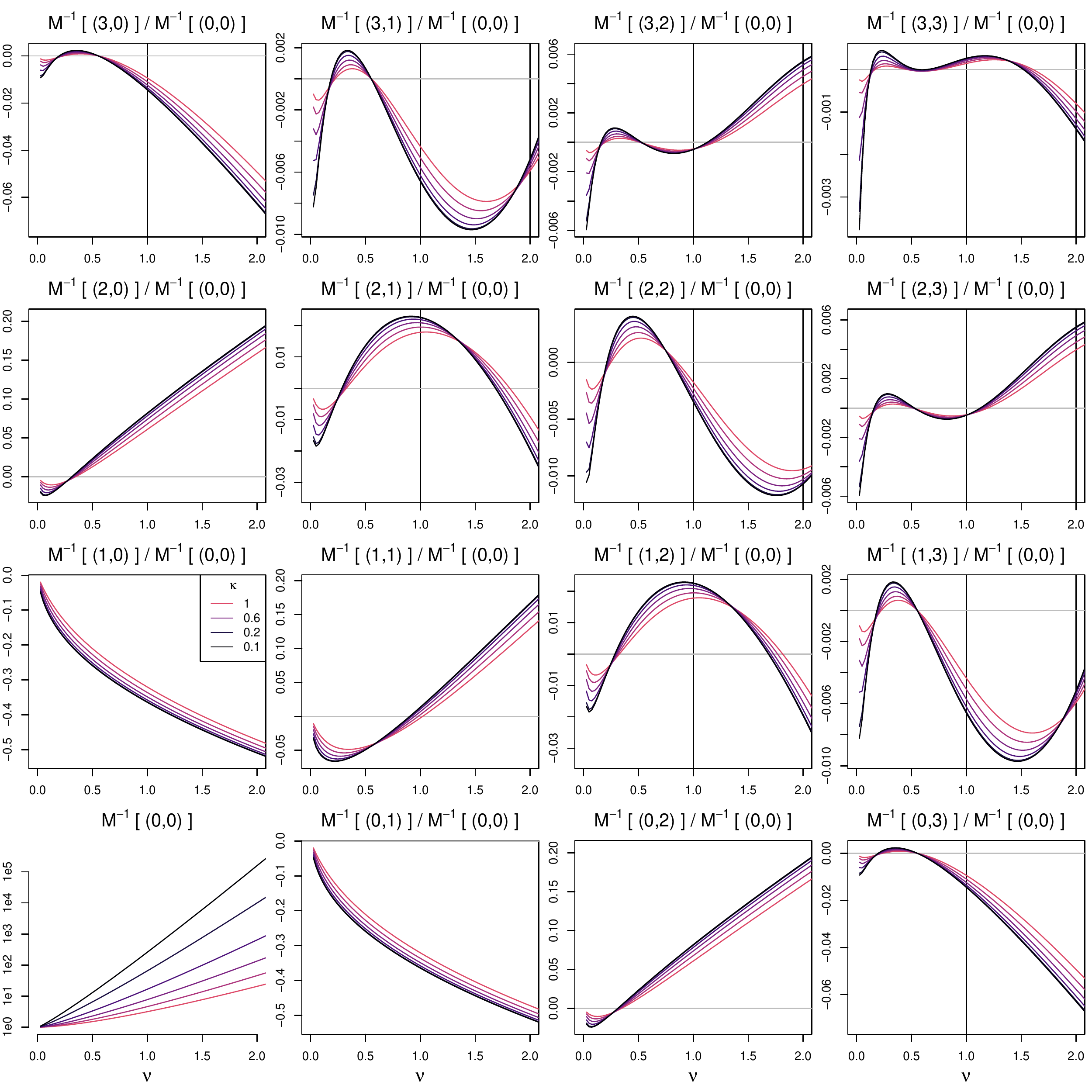}
  \caption{\label{Q_vs_nu_2d} For dimension $d=2$,
    $M_1^{-1}\sqb{h \,;\, \nu, 2}/M_1^{-1}\sqb{0 \,;\, \nu, 2}$ as a function
    of $\nu$ for various values of $h$ and several inverse range
    parameters $\alpha$. Vertical lines indicate values set to zero in
    the SPDE approximation.}
\end{figure}

Parameters were estimated via maximum likelihood under
each of the following scenarios:
\begin{enumerate}
\item true model
\item Vecchia's approximation using GpGp R package.
\item SPDE approximation
\item SPDE approximation double resolution
\end{enumerate}
For the SPDE methods, we ameliorate edge effects by using the
approximations to compute covariances under periodic boundary
conditions on a $(100,100)$ grid and extract the covariances from a
$(30,30)$ subgrid.  Whereas the data grid has $\Delta = 1$, the SPDE
double resolution method computes covariances using $\Delta = 1/2$,
leading to different covariances than those obtained by the SPDE
approximation with $\Delta = 1$. Vecchia's approximation
\citep{vecchia1988estimation} is implemented in the GpGp R package
\citep{GpGp}. GpGp estimates a constant mean parameter and has minor
penalties on small nuggets; see \cite{guinness2019gaussian} for
details.

Figure \ref{simstudy_micro} contains results of the simulation study
over 500 simulation replicates. Each black point in the plot is an
estimate of the microergodic parameter \citep{zhang2004inconsistent}
$\widehat{\sigma}^2 \widehat{\alpha}^2$. The estimates are sorted and
evenly spaced in the horizontal direction, creating a visualization of
the empirical quantile function of the estimates. We include a magenta
line for the true value $2(0.2)^2 = 0.08$.  Several interesting things
arise from the simulation study. Focusing on the zero noise case
first, the SPDE approximations underestimate the microergodic
parameter. This corresponds to an overestimation of the spatial range
since $\alpha$ is an inverse range parameter. This makes sense given
that the SPDE approximation assigns too much power at the highest
frequencies; the likelihood must select a larger range parameter to
dampen the power. The estimates improve somewhat in SPDE Double
Resolution but still have a bias. When noise is added, the SPDE
approximations begin to improve. Vecchia's approximation is accurate
in every case.

\begin{figure}
  \centering
  \includegraphics[width=1.0\textwidth]{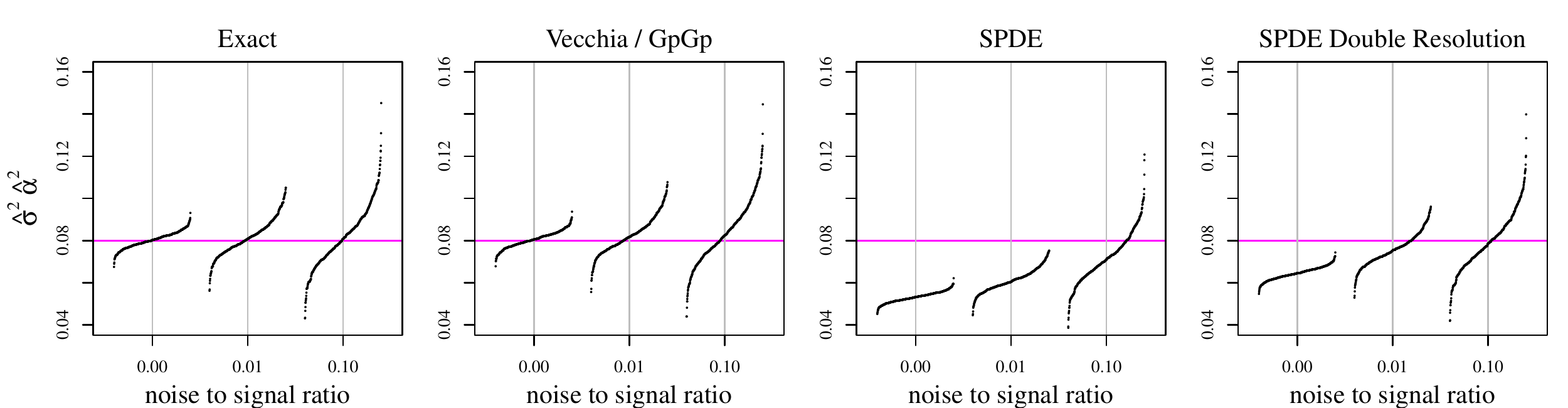}
  \caption{\label{simstudy_micro} Estimates of
    $\widehat{\sigma}^2\widehat{\alpha}^2$, over 500 simulation replicates
    of data on grid of size $(30,30)$ for three noise levels. Each
    black point is a parameter estimate. Estimates are sorted and
    spaced evenly in the horizontal direction to create visualization
    of empirical quantile function. True value
    $\sigma^2 \alpha^2 = 2(0.2)^2 = 0.08$ indicated in magenta.}
\end{figure}

\section{Discussion}\label{discussion}

The SPDE approximation has proven to be useful as a computational tool
and as a conceptual tool for defining extensions to irregularly-spaced
data, models on manifolds, and to non-stationary models
\citep{fuglstad2015exploring,bakka2018spatial}. This paper does not
question the usefulness of the SPDE approach as a tool for data
analysis. Rather, it is a study of the ability of the SPDE approach to
approximate Mat\'ern models on grids.  

We study SPDE approximations to Mat\'ern fields observed at point
locations on a grid, as opposed to observations of gridbox
averages. While SPDE approximations have been applied in both cases,
the spectral properties of gridbox average fields are different, and
it is not clear to the author whether the SPDE approximations would be
more or less accurate in the gridbox average case. This is certainly
and interesting question worthy of future study. In addition, it would
be interesting to explore extensions to irregularly sampled locations.
In both of these cases, a study of the impact of the approximations on
predictions is also warranted. It seems plausible that if the
approximate model is used for both inferring parameters and generating
predictions, the resulting predictions would be reasonably accurate.

\begin{center}
\Large{\bf Acknowledgements}
\end{center}

This work was supported by the National Science Foundation under grant
No.\ 1916208 and the National Institutes of Health under grant No.\
R01ES027892. The author is grateful to have received helpful feedback
on an early draft from Matthias Katzfuss, Michael Stein, and Ethan
Anderes, and from an Associate Editor and two anonymous reviewers, one
of which pointed out an mistake in the first version.

\appendix

\section{ Extended Background }\label{spectral_appendix}

The first proposition establishes the aliasing property of complex
exponential and uses it to express the covariances as an integral on a
bounded domain instead of an infinite domain. Note that we define 
$\T_\Delta = [0,\Delta^{-1}]$ as opposed to $[-\Delta^{-1}/2,\Delta^{-1}/2]$,
which is perhaps a more common convention. Both lead to equivalent
integrals, as all of the aliased functions are periodic on $\T_\Delta$. We
prefer this convention because it more closely maps onto how most
software organizes their fast Fourier transform functions. For example,
in R, the `fft' function returns a vector whose first entry corresponds
to frequency $0$ instead of frequency $-\Delta^{-1}/2$.
\begin{proposition} (Aliasing)
For $h \in \Z^d$,
\begin{align}
  A \sbl \Delta h \sbr
  = \int_{\T_\Delta^d} \Big[ \sum_{k \in \Z^d}
  \mcA(\omega + k/\Delta) \Big] \exp({i2\pi \Delta \omega \cdot h})d\omega
\end{align}
\end{proposition}
\begin{proof}
Using Bochner's theorem and splitting the integral into domains of size $\T_\Delta^d$,
\begin{align}
A \sqb{\Delta h} = 
\sum_{k \in \Z^d} 
\int_{\T^d_\Delta} A(\omega + k/\Delta) \exp(i2\pi (\omega+k/\Delta) \cdot \Delta h) d\omega.
\end{align}
Exchanging sum with integral gives
\begin{align}
A \sqb{\Delta h} = 
 \int_{\T^d_\Delta} 
\sum_{k \in \Z^d}
A(\omega + k/\Delta) \exp(i2\pi (\omega+k/\Delta) \cdot \Delta h) d\omega.
\end{align}
The complex exponential can be expanded as
\begin{align}
\exp(i 2\pi(\omega + k/\Delta) \cdot \Delta h) 
= \exp(i 2\pi \Delta \omega \cdot h + i2\pi k \cdot h)
= \exp(i 2\pi \Delta \omega \cdot h)\exp(i 2\pi k \cdot h ).
\end{align}
Since $k \in \Z^2$ and $h\in \Z^2$, $k \cdot h = n$ is an integer, and
thus the complex exponential is
\begin{align}
\exp(i2\pi \Delta \omega \cdot h) \exp(i 2\pi n) = \exp(i 2\pi \Delta \omega \cdot h).
\end{align}
Plugging this expression back into the integral gives
\begin{align}
A \sqb{\Delta h} = \int_{\T^d_\Delta} 
\Big[ \sum_{k \in \Z^d} A( \omega + k/\Delta ) \Big] 
\exp( i 2\pi \Delta \omega \cdot h ) d\omega,
\end{align}
as desired.
\end{proof}

The second proposition establishes the reciprocal relationship between
the spectral density of the covariance operator and the spectral
density of the inverse operator.

\begin{proposition}
(Reciprocal Relationship of Inverse) For all $\omega \in \T_\Delta^d$,
\begin{align}
A_\Delta(\omega) A_\Delta^{-1}(\omega) = \Delta^d .
\end{align}
\end{proposition}

\begin{proof}
By definition, the inverse satisfies
\begin{align}
\Delta^d \sum_{k \in \Z^d} A_\Delta\sqb{h-k} A_\Delta^{-1}\sqb{k} = \mathbbm{1}\sqb{h}.
\end{align}
Take the infinite DFT of both sides,
\begin{align}
\Delta^{2d} \sum_{h \in \Z^d} \sum_{k \in \Z^d} 
A_\Delta \sqb{h-k} A_\Delta^{-1}\sqb{k}e^{-i2\pi\Delta\omega \cdot h}
&= \Delta^d \sum_{h \in \Z^d} \mathbbm{1}\sqb{h}e^{-i2\pi\Delta\omega \cdot h} \\  
 \Delta^d \sum_{k \in \Z^d} A_\Delta^{-1}\sqb{k} e^{-i2\pi\Delta\omega \cdot k} 
 \Delta^d\sum_{h \in \Z^d} A_\Delta\sqb{h-k} e^{-i2\pi\Delta\omega \cdot (h-k)} 
&= \Delta^d \\
A_\Delta^{-1}(\omega) A_\Delta(\omega) &= \Delta^d.
\end{align}
\end{proof}

\begin{proposition}(Square Root spectral density)
For all $\omega \in \T_\Delta^d$,
\begin{align}
A_\Delta^{1/2}(\omega) A_\Delta^{1/2}(\omega)^* = A_\Delta(\omega).
\end{align}
\end{proposition}
\begin{proof}
By definition, the square root operator satisfies
\begin{align}
\Delta^d \sum_{k \in \Z^d} A_\Delta^{1/2} \sqb{h-k} A_\Delta^{1/2}\sqb{-k} = A_\Delta\sqb{h}.
\end{align}
Taking the infinite DFT of both sides,
\begin{align}
\Delta^{2d} \sum_{h \in \Z^d} \sum_{k \in \Z^d} A_\Delta^{1/2} \sqb{h-k} A_\Delta^{1/2}\sqb{-k} 
e^{-i2\pi\Delta\omega \cdot h}
&= \Delta^d \sum_{h \in \Z^d} A_\Delta\sqb{h}e^{-i2\pi\Delta\omega \cdot h} \\
 \Delta^d \sum_{k \in \Z^d} A_\Delta^{1/2}\sqb{-k} e^{-i2\pi\Delta\omega \cdot k} 
 \Delta^{d} \sum_{h \in \Z^d}  A_\Delta^{1/2} \sqb{h-k} e^{-i2\pi\Delta\omega \cdot (h-k)} 
&= A_\Delta(\omega) \\
A_\Delta^{1/2}(\omega)^* A_\Delta^{1/2}(\omega) &= A_\Delta(\omega).
\end{align}
\end{proof}

\begin{proposition} (Convolution Method of Simulation)
If the mean-zero process $W : \Z^d \to \R$ has covariance operator $\mathbbm{1}\sqb{h}$, and
\begin{align}
Y\sqb{\Delta h} = \Delta^{d/2} \sum_{k\in\Z^d} A_\Delta^{1/2} \sqb{h-k} W\sqb{k},
\end{align}
then $\Cov(Y\sqb{\Delta(h)}, Y\sqb{\Delta j}) = A_\Delta\sqb{h-j}$.
\end{proposition}
\begin{proof}
\begin{align}
\Cov( Y\sqb{\Delta h}, Y\sqb{\Delta j} ) 
&= E \Big[ \Delta^{d/2} \sum_{k \in \Z^d} A_\Delta^{1/2}\sqb{h-k}W\sqb{k} 
           \Delta^{d/2} \sum_{m \in \Z^d} A_\Delta^{1/2}\sqb{j-m}W\sqb{m} \Big] \\
&= \Delta^d \sum_{k \in \Z^d} A_\Delta^{1/2}\sqb{h-k} A_\Delta^{1/2}\sqb{j-k} \\
&= \Delta^d \sum_{\ell \in \Z^d} A_\Delta^{1/2}\sqb{h-j-\ell}A_\Delta^{1/2}\sqb{-\ell} \\
&= A_\Delta\sqb{h-j}.
\end{align}
\end{proof}

\begin{proposition}
If the mean-zero process $Y : (\Delta \Z)^d \to \R$ has covariance operator $A_\Delta$, and
\begin{align}
W\sqb{h} = \Delta^{d/2} \sum_{k \in \Z^d} A_\Delta^{-1/2}\sqb{h-k} Y\sqb{\Delta k},
\end{align}
then $\Cov( W\sqb{h}, W\sqb{j} ) = \mathbbm{1}\sqb{h-j}$.
\end{proposition}
\begin{proof}
\begin{align}
\Cov( &W\sqb{h}, W\sqb{j} ) \\ 
&= E\Big[ \Delta^{d/2} \sum_{k \in \Z^d} A_\Delta^{-1/2}\sqb{h-k} Y\sqb{\Delta k}
          \Delta^{d/2} \sum_{m \in \Z^d} A_\Delta^{-1/2}\sqb{j-m} Y\sqb{\Delta m} \Big] \\
&= \Delta^d \sum_{k \in \Z^d} \sum_{m \in \Z^d} 
A_\Delta^{-1/2}\sqb{h-k} A_\Delta^{-1/2}\sqb{j-m} A_\Delta\sqb{m-k} \\
&= \Delta^d \sum_{k \in \Z^d} \sum_{m \in \Z^d} 
A_\Delta^{-1/2}\sqb{h-k} A_\Delta^{-1/2}\sqb{j-m} 
\Delta^d \sum_{\ell \in \Z^d} A_\Delta^{1/2}\sqb{m-k-\ell} A_\Delta^{1/2}\sqb{-\ell} \\
&= \Delta^d \sum_{k \in \Z^d} \sum_{\ell \in \Z^d} 
A_\Delta^{-1/2}\sqb{h-k}  A_\Delta^{1/2}\sqb{-\ell}
\Delta^d \sum_{m \in \Z^d} A_\Delta^{1/2}\sqb{m-k-\ell} A_\Delta^{-1/2}\sqb{j-m} \\
&= \Delta^d \sum_{k \in \Z^d} \sum_{\ell \in \Z^d} 
A_\Delta^{-1/2}\sqb{h-k}  A_\Delta^{1/2}\sqb{-\ell}
\Delta^d \sum_{n \in \Z^d} A_\Delta^{1/2}\sqb{j-k-\ell-n} A_\Delta^{-1/2}\sqb{n} \\
&= \Delta^d \sum_{k \in \Z^d} \sum_{\ell \in \Z^d} 
A_\Delta^{-1/2}\sqb{h-k}  A_\Delta^{1/2}\sqb{-\ell} 
\mathbbm{1}\sqb{j-k-\ell} \\
&= \Delta^d \sum_{k \in \Z^d} 
A_\Delta^{-1/2}\sqb{h-k}  A_\Delta^{1/2}\sqb{k-j} \\
&= \mathbbm{1}\sqb{h-j}
\end{align}
\end{proof}

\setcounter{theorem}{0}

\section{Proofs for Mat\'ern Model}\label{proofs}

This section contains proofs of Theorems 1-4 from the main
document. Theorems 1-3 have precise coefficients on higher-order
terms that were omitted from the main document to conserve space.  We
first state a lemma that is a consequence of Taylor's Theorem:
\begin{lemma} For $x > 0$, 
\begin{align*}
(1 + x)^{-1} &= 1 - x + x^2 c(x), \quad 0 < c(x) < 1 \\
(1 + x)^{-2} &= 1 - 2x + x^2 d(x), \quad 0 < d(x) < 3
\end{align*}
\end{lemma}
\begin{proof}
By Taylor's Theorem,
$(1 + x)^{-1} = 1 - x + x^2 (1+a(x))^{-3}$,
where $0 < a(x) < x$. Therefore,
$(1 + x)^{-1} = 1 - x + x^2 c(x), \mbox{where} \quad 0 < c(x) < 1$.
Similarly, $(1 + x)^{-2} = 1 - 2 x + 3 x^2 (1+b(x))^{-4}$,
where $0 < b(x) < x$. Therefore,
$(1 + x)^{-2} = 1 - 2 x + x^2 d(x), \mbox{where} \quad 0 < d(x) < 3$.
\end{proof}

\begin{theorem}\label{specden_1o2_app}
\begin{align*}
\frac{M_\Delta( 0 \,;\, 1/2, 1 )}{2/\alpha} &= 1 + O(\alpha^2 \Delta^2)  
&\frac{M_\Delta\big( \frac{1}{2\Delta} \,;\, 1/2, 1 \big)}{2/\alpha} &= \frac{\alpha^2\Delta^2}{4} + O(\alpha^4 \Delta^4) \\
\frac{\Mt_\Delta( 0 \,;\, 1/2, 1 )}{2/\alpha} &=  1 
&\frac{\Mt_\Delta\big( \frac{1}{2\Delta} \,;\, 1/2, 1 \big)}{2/\alpha} &= \frac{\alpha^2\Delta^2}{4} + O(\alpha^4\Delta^4).
\end{align*}
\end{theorem}

\begin{proof}
We establish each of the four relations in turn. Rearranging terms gives
\begin{align*}
\frac{M_\Delta(0 : 1/2, 1)}{2/\alpha} = 
    \alpha^2 \sum_{k \in \Z} \Big( \alpha^2 + 4\pi^2k^2/\Delta^2 \Big)^{-1} 
  = 1 + 
    \frac{\alpha^2 \Delta^2}{2\pi^2} 
    \sum_{k=1}^{\infty} \frac{1}{k^2} 
    \Big( 1 + \frac{\alpha^2 \Delta^2}{4 \pi^2 k^2} \Big)^{-1}.
\end{align*}
Applying Lemma 1,
\begin{align*}
\frac{M_\Delta(0 : 1/2, 1)}{2/\alpha} 
 = 1 + 
    \frac{\alpha^2 \Delta^2}{2\pi^2} 
    \sum_{k=1}^{\infty} \frac{1}{k^2} 
    \Big( 1 - \frac{\alpha^2 \Delta^2}{4 \pi^2 k^2} 
    + q_k\frac{\alpha^4 \Delta^4}{16 \pi^4 k^4} \Big),
\end{align*}
where $0 < q_k < 1$. Evaluating the sums gives
\begin{align*}
\frac{M_\Delta(0 : 1/2, 1)}{2/\alpha} 
 = 1 + 
    \frac{\alpha^2 \Delta^2}{12} 
  - \frac{\alpha^4 \Delta^4}{720}
  + O(\alpha^6 \Delta^6).
\end{align*}

The second relation follows directly from plugging $\omega = 0$ into the 
SPDE spectral density. For the third relation, rearranging terms gives
\begin{align*}
\frac{M_\Delta( \Delta^{-1}/2 : 1/2, 1 )}{2/\alpha} 
    &= \alpha^2 \sum_{k \in \Z} 
      \Big[ \alpha^2 + 4\pi^2 \Big(\frac{1}{2\Delta} + \frac{k}{\Delta} \Big)^2 \Big]^{-1}\\
    &= \frac{2\alpha^2\Delta^2}{\pi^2} 
      \sum_{k=0}^\infty 
      \frac{1}{(2k + 1)^2}
       \Big[ 1 + \frac{\alpha^2 \Delta^2}{\pi^2} \frac{1}{(2k + 1 )^2} \Big]^{-1}.
\end{align*}
Applying Lemma 1,
\begin{align*}
\frac{M_\Delta( \Delta^{-1}/2 : 1/2, 1 )}{2/\alpha} 
  = \frac{2\alpha^2\Delta^2}{\pi^2}
    \sum_{k=0}^\infty 
    \frac{1}{(2k+1)^2}\Big[ 1 - \frac{\alpha^2\Delta^2}{\pi^2} \frac{1}{(2k+1)^2} + 
    p_k \frac{\alpha^4\Delta^4}{\pi^4} \frac{1}{(2k+1)^4} \Big],
\end{align*}
where $0 < p_k < 1$. Evaluating the sums,
\begin{align*}
\frac{M_\Delta( \Delta^{-1}/2 : 1/2, 1 )}{2/\alpha} 
  = \frac{\alpha^2\Delta^2}{4}
  - \frac{\alpha^4\Delta^4}{48}
  + O(\alpha^6 \Delta^6),
\end{align*}
establishing the third relation. For the fourth relation,
\begin{align*}
\frac{\Mt( \Delta^{-1}/2 : 1/2, 1 ) }{2/\alpha} 
  = \frac{\alpha \Delta}{2}\Big( \frac{\alpha\Delta}{2} + \frac{2}{\alpha\Delta} \Big)^{-1}
  = \frac{\alpha^2\Delta^2}{4}\Big( 1 + \frac{\alpha^2\Delta^2}{4} \Big)^{-1}.
\end{align*}
Applying Lemma 1, for $0 < q < 1$,
\begin{align*}
\frac{\Mt( \Delta^{-1}/2 : 1/2, 1 ) }{2/\alpha} 
  = \frac{\alpha^2\Delta^2}{4}\Big( 1 - \frac{\alpha^2\Delta^2}{4} 
  + q \frac{\alpha^4\Delta^4}{16}  \Big) 
  = \frac{\alpha^2\Delta^2}{4} - \frac{\alpha^4\Delta^4}{16} + O(\alpha^6\Delta^6).
\end{align*}

\end{proof}

\begin{theorem}\label{specden_3o2_app}
\begin{align*}
\frac{M_\Delta( 0 \,;\, 3/2, 1 )}{4/\alpha} &= 
    1 + O(\alpha^4 \Delta^4) 
&\frac{M_\Delta\big( \frac{1}{2\Delta} \,;\, 3/2, 1 \big)}{4/\alpha} &=
    \frac{\alpha^4\Delta^4}{48} + O(\alpha^6 \Delta^6)  \\ 
\frac{\Mt_\Delta( 0 \,;\, 3/2, 1 )}{4/\alpha} &= 1 
&\frac{\Mt_\Delta\big( \frac{1}{2\Delta} \,;\, 3/2, 1 \big)}{4/\alpha} &=
    \frac{\alpha^4\Delta^4}{16} + O(\alpha^6\Delta^6).
\end{align*}
\end{theorem}

\begin{proof}
We establish each of the four relations in turn. Rearranging terms gives
\begin{align*}
\frac{M_\Delta(0 : 3/2, 1)}{4/\alpha} = 
    \alpha^4 \sum_{k \in \Z} \Big( \alpha^2 + 4\pi^2k^2/\Delta^2 \Big)^{-2} 
  = 1 + 
    \frac{\alpha^4 \Delta^4}{8\pi^4} 
    \sum_{k=1}^{\infty} \frac{1}{k^4} 
    \Big( 1 + \frac{\alpha^2 \Delta^2}{4 \pi^2 k^2} \Big)^{-2}.
\end{align*}
Applying Lemma 1,
\begin{align*}
\frac{M_\Delta(0 : 3/2, 1)}{4/\alpha} 
 = 1 + 
    \frac{\alpha^4 \Delta^4}{8\pi^4} 
    \sum_{k=1}^{\infty} \frac{1}{k^4} 
    \Big( 1 - \frac{\alpha^2 \Delta^2}{2 \pi^2 k^2} 
     + \frac{\alpha^4 \Delta^4}{16 \pi^4 k^4} q_k 
    \Big).
\end{align*}
where $0 < q_k < 3$. Evaluating the sums gives
\begin{align*}
\frac{M_\Delta(0 : 3/2, 1)}{2/\alpha} 
 = 1 + 
    \frac{\alpha^4 \Delta^4}{720} 
  - \frac{\alpha^6 \Delta^6}{15120}
  + O(\alpha^8 \Delta^8).
\end{align*}

The second relation follows directly from plugging $\omega = 0$ into the 
SPDE spectral density. For the third relation, rearranging terms gives
\begin{align*}
\frac{M_\Delta( \Delta^{-1}/2 : 3/2, 1 )}{4/\alpha} 
    &= \alpha^4 \sum_{k \in \Z} 
      \Big[ \alpha^2 + 4\pi^2 \Big(\frac{1}{2\Delta} + \frac{k}{\Delta} \Big)^2 \Big]^{-2}\\
    &= \frac{2\alpha^4\Delta^4}{\pi^4} 
      \sum_{k=0}^\infty 
      \frac{1}{(2k + 1)^4}
       \Big[ 1 + \frac{\alpha^2 \Delta^2}{\pi^2} \frac{1}{(2k + 1 )^2} \Big]^{-2}.
\end{align*}
Applying Lemma 1,
\begin{align*}
\frac{M_\Delta( \Delta^{-1}/2 : 3/2, 1 )}{4/\alpha} 
    &= \frac{2\alpha^4\Delta^4}{\pi^4} 
      \sum_{k=0}^\infty 
      \frac{1}{(2k + 1)^4}
       \Big[ 1 - \frac{2\alpha^2 \Delta^2}{\pi^2} \frac{1}{(2k + 1 )^2} 
           + p_k \frac{\alpha^4\Delta^4}{\pi^4}\frac{1}{(2k+1)^4} \Big]
\end{align*}
where $0 < p_k < 3$. Evaluating the sums,
\begin{align*}
\frac{M_\Delta( \Delta^{-1}/2 : 3/2, 1 )}{4/\alpha} 
  = \frac{\alpha^4\Delta^4}{48}
  - \frac{\alpha^6\Delta^6}{240}
  + O(\alpha^8 \Delta^8),
\end{align*}
establishing the third relation. For the fourth relation,
\begin{align*}
\frac{\Mt( \Delta^{-1}/2 : 3/2, 1 ) }{4/\alpha} 
  = \frac{\alpha^2 \Delta^2}{4}\Big( \frac{\alpha\Delta}{2} + \frac{2}{\alpha\Delta} \Big)^{-2}
  = \frac{\alpha^4\Delta^4}{16}\Big( 1 + \frac{\alpha^2\Delta^2}{4} \Big)^{-2}.
\end{align*}
Applying Lemma 1, for $0 < q < 3$,
\begin{align*}
\frac{\Mt( \Delta^{-1}/2 : 3/2, 1 ) }{4/\alpha} 
  = \frac{\alpha^4\Delta^4}{16}\Big( 1 - \frac{\alpha^2\Delta^2}{2} 
  + q \frac{\alpha^4\Delta^4}{16}  \Big) 
  = \frac{\alpha^4\Delta^4}{16} - \frac{\alpha^6\Delta^6}{32} + O(\alpha^8\Delta^8).
\end{align*}
\end{proof}

\begin{theorem}\label{theorem_2d_true_app}
\begin{align}
\frac{M_\Delta( (0,0) \,;\, 1, 2 )}{4\pi/\alpha^2} &=
 1 +  \frac{\alpha^4 \Delta^4}{258.6} + O(\alpha^6 \Delta^6)\\
\frac{ M_\Delta\big( (\frac{1}{2\Delta},0)\,;\, 1,2 \big) }{4\pi/\alpha^2} &=
  \frac{\alpha^4\Delta^4}{43.10} + O(\alpha^6\Delta^6)\\
\frac{ M_\Delta\big( (\frac{1}{2\Delta},\frac{1}{2\Delta})\,;\, 1,2 \big) }{4\pi/\alpha^2} &=
 \frac{\alpha^4\Delta^4}{86.20} + O(\alpha^6\Delta^6)
\end{align}
\end{theorem}

\begin{proof}

For the first relationship,
\begin{align*}
\frac{M_\Delta( (0,0) : 1, 2 )}{4\pi/\alpha^2} &=
  \alpha^4 \sum_{k \in \Z^2} 
  \Big( \alpha^2 + 4\pi^2k_1^2/\Delta^2 + 4\pi^2 k_2^2/\Delta^2 \Big)^{-2} \\
  &=
  1 + \alpha^4 \sum_{k \neq (0,0)} 
  \Big( \alpha^2 + 4\pi^2k_1^2/\Delta^2 + 4\pi^2 k_2^2/\Delta^2 \Big)^{-2} \\
  &= 1 + \frac{\alpha^4\Delta^4}{16\pi^4}
     \sum_{k \neq (0,0) } 
     (k_1^2 + k_2^2)^{-2}
     \bigg( 1 + \frac{\alpha^2\Delta^2}{4\pi^2} \frac{1}{k_1^2 + k_2^2} \bigg)^{-2}
\end{align*}
Applying Lemma 1,
\begin{align*}
\frac{M_\Delta( (0,0) : 1, 2 )}{4\pi/\alpha^2} &=
   1 + \frac{\alpha^4\Delta^4}{16\pi^4}
     \sum_{k \neq (0,0) } 
     (k_1^2 + k_2^2)^{-2}
     \bigg[ 1 - \frac{\alpha^2\Delta^2}{2\pi^2} \frac{1}{k_1^2 + k_2^2} 
     + q_{k_1,k_2}\frac{\alpha^4\Delta^4}{16\pi^4} \frac{1}{(k_1^2 + k_2^2)^2},
     \bigg]
\end{align*}
for $0 < q_{k_1,k_2} < 3$. Evaluating the sums numerically,
\begin{align*}
\frac{M_\Delta( (0,0) : 1, 2 )}{4\pi/\alpha^2} &=
   1 + \frac{\alpha^4\Delta^4}{258.602} + \frac{\alpha^6\Delta^6}{6603.35}
     + O(\alpha^8\Delta^8),
\end{align*}
proving the first relationship. For the second relationship,
\begin{align*}
\frac{M_\Delta( (\frac{1}{2\Delta},0) : 1, 2 )}{4\pi/\alpha^2} &=
  \alpha^4 \sum_{k \in \Z^2} 
  \Big[ \alpha^2 + 4\pi^2\Big(\frac{1}{2\Delta} + \frac{k_1}\Delta\Big)^2 
  + 4\pi^2 k_2^2/\Delta^2 \Big]^{-2} \\
  &=
  \frac{\alpha^4\Delta^4}{\pi^4} \sum_{k \in \Z^2} 
  \Big[ (2k_1 + 1)^2 + (2k_2)^2 \Big]^{-2} 
  \bigg[ 1 + 
  \frac{\alpha^2\Delta^2}{\pi^2} \frac{1}{(2k_1 + 1)^2 + (2k_2)^2} 
  \bigg]^{-2}
\end{align*}
Applying Lemma 1,
\begin{align*}
\frac{M_\Delta( (\frac{1}{2\Delta},0) : 1, 2 )}{4\pi/\alpha^2} 
&=
  \frac{\alpha^4\Delta^4}{\pi^4} \sum_{k \in \Z^2} 
  \Big[ (2k_1 + 1)^2 + (2k_2)^2 \Big]^{-2} \\
& \hspace{20mm} \bigg[ 1 - 
  \frac{2\alpha^2\Delta^2/\pi^2}{(2k_1 + 1)^2 + (2k_2)^2} 
+ q_{k_1,k_2}\frac{\alpha^4\Delta^4/\pi^4}{[(2k_1 + 1)^2 + (2k_2)^2]^2}
  \bigg]^{-2}
\end{align*}
for $0 < q_{k_1,k_2} < 3$. Evaluating the sums numerically,
\begin{align*}
\frac{M_\Delta( (\frac{1}{2\Delta},0) : 1, 2 )}{4\pi/\alpha^2} 
&=
   \frac{\alpha^4\Delta^4}{43.1003} + \frac{\alpha^6\Delta^6}{23.8950}
     + O(\alpha^8\Delta^8),
\end{align*}
proving the second relationship. For the third relationship,
\begin{align*}
\frac{M_\Delta( (\frac{1}{2\Delta},\frac{1}{2\Delta}) : 1, 2 )}{4\pi/\alpha^2} 
&=
  \alpha^4 \sum_{k \in \Z^2} 
  \Big[ 
  \alpha^2 
  + 
  4\pi^2\Big(\frac{1}{2\Delta} + \frac{k_1}\Delta\Big)^2 
  + 
  4\pi^2\Big(\frac{1}{2\Delta} + \frac{k_2}\Delta\Big)^2 
  \Big] \\
&=
  \frac{\alpha^4\Delta^4}{\pi^4} \sum_{k \in \Z^2} 
  \Big[ (2k_1 + 1)^2 + (2k_2)^2 \Big]^{-2} 
  \bigg[ 1 + 
  \frac{\alpha^2\Delta^2}{\pi^2} \frac{1}{(2k_1 + 1)^2 + (2k_2+1)^2} 
  \bigg]^{-2}
\end{align*}
Applying Lemma 1,
\begin{align*}
\frac{M_\Delta( (\frac{1}{2\Delta},\frac{1}{2\Delta}) : 1, 2 )}{4\pi/\alpha^2} 
&=
  \frac{\alpha^4\Delta^4}{\pi^4} \sum_{k \in \Z^2} 
  \Big[ (2k_1 + 1)^2 + (2k_2+1)^2 \Big]^{-2} \\
& \hspace{20mm} \bigg[ 1 - 
  \frac{2\alpha^2\Delta^2/\pi^2}{(2k_1 + 1)^2 + (2k_2+1)^2} 
+ q_{k_1,k_2}\frac{\alpha^4\Delta^4/\pi^4}{[(2k_1 + 1)^2 + (2k_2+1)^2]^2}
  \bigg]^{-2}
\end{align*}
for $0 < q_{k_1,k_2} < 3$. Evaluating the sums numerically,
\begin{align*}
\frac{M_\Delta( (\frac{1}{2\Delta},\frac{1}{2\Delta}) : 1, 2 )}{4\pi/\alpha^2} 
&=
   \frac{\alpha^4\Delta^4}{86.2007} + \frac{\alpha^6\Delta^6}{95.5799}
     + O(\alpha^8\Delta^8),
\end{align*}
proving the third relationship.
\end{proof}

\begin{theorem}\label{theorem_2d_SPDE_app}
\begin{align}
\frac{\Mt_\Delta( (0,0) \,;\, 1, 2 )}{4\pi/\alpha^2} 
&= 1  \\
\frac{\Mt_\Delta\big( (\frac{1}{2\Delta},0)\,;\, 1,2 \big) }{4\pi/\alpha^2} 
&= \frac{\alpha^4\Delta^4}{16}  + O( \alpha^6\Delta^6 )\\
\frac{\Mt_\Delta\big( (\frac{1}{2\Delta},\frac{1}{2\Delta})\,;\, 1,2 \big) }{4\pi/\alpha^2} 
&= \frac{\alpha^4\Delta^4}{64} + O( \alpha^6\Delta^6 )
\end{align}
\end{theorem}

\begin{proof}
  The first relationship follows directly from plugging
  $\omega = (0,0)$ into the formula for the spectral density.

For the second relationship,
\begin{align*}
\frac{\Mt_\Delta( (\frac{1}{2\Delta},0) : 1, 2 )}{4\pi/\alpha^2} 
&=
\alpha^4\Delta^4 \Big( 4 + \alpha^2\Delta^2 \Big)^{-2}
= 
\frac{\alpha^4\Delta^4}{16} \Big( 1 + \frac{\alpha^2\Delta^2}{4} \Big)^{-2}.
\end{align*}
Applying Lemma 1,
\begin{align*}
\frac{\Mt_\Delta( (\frac{1}{2\Delta},0) : 1, 2 )}{4\pi/\alpha^2} 
= 
\frac{\alpha^4\Delta^4}{16} 
\Big( 
1 - \frac{\alpha^2\Delta^2}{2} + q\frac{\alpha^4\Delta^4}{16} \Big)
=
\frac{\alpha^4\Delta^4}{16} + \frac{\alpha^6\Delta^6}{32} + O(\alpha^8\Delta^8),
\end{align*}
where $0 < q < 3$, proving the second relationship. For the third relationship,
\begin{align*}
\frac{\Mt_\Delta( (\frac{1}{2\Delta},\frac{1}{2\Delta}) : 1, 2 )}{4\pi/\alpha^2} 
&=
\alpha^4\Delta^4 \Big( 8 + \alpha^2\Delta^2 \Big)^{-2}
= 
\frac{\alpha^4\Delta^4}{64} \Big( 1 + \frac{\alpha^2\Delta^2}{8} \Big)^{-2}.
\end{align*}
Applying Lemma 1,
\begin{align*}
\frac{\Mt_\Delta( (\frac{1}{2\Delta},\frac{1}{2\Delta}) : 1, 2 )}{4\pi/\alpha^2} 
= 
\frac{\alpha^4\Delta^4}{64} 
\Big( 
1 - \frac{\alpha^2\Delta^2}{4} + q\frac{\alpha^4\Delta^4}{64} \Big)
=
\frac{\alpha^4\Delta^4}{64} + \frac{\alpha^6\Delta^6}{256} + O(\alpha^8\Delta^8),
\end{align*}
where $0 < q < 3$, proving the third relationship.
\end{proof}

\bibliography{../refs}{}
\bibliographystyle{apalike}

\end{document}